\documentclass[12pt,a4paper]{article}

\usepackage{amsmath,amssymb,amsthm}
\usepackage{booktabs}
\usepackage{graphicx}
\usepackage{geometry}
\usepackage{hyperref}
\hypersetup{colorlinks=true,linkcolor=blue,citecolor=blue,urlcolor=blue}
\usepackage[numbers]{natbib}
\usepackage{enumitem}
\usepackage{makecell}

\newtheorem{theorem}{Theorem}[section]
\newtheorem{lemma}[theorem]{Lemma}
\newtheorem{proposition}[theorem]{Proposition}

\newtheorem{definition}{Definition}[section]

\newtheorem{assumption}{Assumption}[section]

\newcommand{\R}{\mathbb{R}}
\newcommand{\p}{\mathbf{p}}
\newcommand{\w}{\mathbf{w}}
\newcommand{\cvec}{\mathbf{c}}

\begin{document}
	
	% ---------- Title and Abstract ----------
	\title{A Geometric Approach to the Transformation Problem of Values}
	\author{Jiyuan Lyu\thanks{Jiyuan Lyu, Postgraduate, School of Economics and Management,Beijing University of Technology,Email: \href{mailto:lyujiyuan@emails.bjut.edu.cn}{sergeylyu@163.com}}}
	\date{\today}
	\maketitle
	
	\begin{abstract}
		The reduction of complex labour to simple labour is an unresolved difficulty in Marx's labour theory of value, and a key obstacle that has prevented the transformation problem from being settled definitively. This paper proposes a two-step solution framework. First, we prove that as long as the macroeconomy generates a physical surplus, the reduction coefficients that respect the floor of labour-power reproduction form a bounded ``value feasible region''; within this region the two macro aggregate equalities can hold simultaneously for a reasonable range of the profit rate. Second, we propose a linear mapping method that exploits the observable structure of nominal wages and the reproduction floor constraint to systematically construct the implicit reduction coefficients from the value feasible region. We show that this mapping is a homeomorphism between the price feasible region and the value feasible region, and that it preserves the boundary structure. An empirical calibration based on China's 2017 inter-provincial input--output table with 1272 sectors shows that the reduction coefficients obtained by the mapping method substantially outperform the homogeneous labour method and the wage-proxy method in matching the macro profit share. 
		
		\vspace{1ex}
		
	\end{abstract}
	
	% ================================================================
	\section{Introduction}
	% ================================================================
	
	The transformation problem of values---how commodity values are transformed into prices of production and how surplus value is transformed into average profit---is one of the most persistent and profound theoretical debates in Marxian political economy. Since the publication of Volume III of \textit{Capital} in 1894, scholars have debated for more than a century whether Marx's transformation procedure is logically consistent and whether the two macro aggregate equalities (total price equals total value, total profit equals total surplus value) can hold simultaneously.
	
	Since Bortkiewicz's classic critique of Volume III, the transformation problem has been among the most controversial issues in Marxian economics. With the rigorous formalisation of the Sraffian system and input--output analysis, many studies have argued that, given physical technology and the real wage, prices of production and the profit rate can be directly determined by a ``production of commodities by means of commodities'' system, so that ``value'' appears to be a redundant intermediary \cite{sraffa_production_2016, kurz_theory_1997, pasinetti_lectures_1977, samuelson_understanding_1971, steedman_heterogeneous_2019}. In contrast, the mathematical Marxian literature has tried to provide a more rigorous expression of the labour theory of value, for example by restating the classical problem with the tools of convex analysis, linear algebra and non-negative matrix theory \cite{roemer1981analytical, flaschel2010topics, morishima_marxs_1973}. Among these efforts, one important line attempts to derive an objective vector of values from the physical network; but in a static setting this line usually finds it difficult to satisfy simultaneously the two macro equalities of total price equals total value and total profit equals total surplus value. Another line resolves the contradiction by redefining the relationship between value and money, or by introducing a temporal structure, as in the ``New Interpretation'' and the ``Temporal Single System Interpretation'' \cite{foley_value_1982, dumenil_beyond_1983, dumenil_labor_1989, kliman_temporal_1999, veneziani_dynamics_2005}.
	
	Whether it is the Bortkiewicz tradition, Morishima's iterative solution, the New Interpretation, or the TSSI, almost all mainstream solution attempts implicitly accept the same premise: the labour of different sectors is homogeneous, or can be uniquely converted into homogeneous simple labour by a vector of reduction coefficients given in advance. The classical solutions and Morishima's model take homogeneous labour as the starting point; although the New Interpretation discusses the endogenous determination of the value of labour power, it directly equates the wage share in total value added with the value of labour power, which still implies the assumption of labour homogeneity; the TSSI likewise does not systematically model labour heterogeneity. In other words, in these discussions the problem of ``how complex labour is reduced to simple labour'' has been set aside in the mathematical modelling of the transformation problem.
	
	If the reduction coefficients are treated as uniquely fixed numbers (and often further simplified to all equal one), then the two aggregate equalities can indeed hold only under extremely restrictive conditions. However, once one recognises that the reduction coefficients themselves are also products of social reproduction, they are no longer a single point but a set bounded by objective reproduction conditions. This is precisely the starting point of the present paper in shifting the research perspective. Even if one admits theoretically that a feasible range exists for the reduction coefficients, previous research has not answered a more practical question: in the mass of economic data, how should one systematically find the set of reduction coefficients that is ``selected'' by the social process? In other words, existing work either remains at the level of an existence proof, or directly uses market wages as a proxy for the reduction coefficients. The latter is operationally convenient, but methodologically it falls into the circular reasoning of using an outcome of the price system to illustrate the premises of the value system.
	
	Building on earlier work, this paper advances the solution in two steps. The first step relaxes the reduction coefficients from a fixed point to a feasible region, and proves that within this feasible region the two macro aggregate equalities hold as a normal case rather than a special one; this answers the question of existence. The second step constructs a linear mapping that uses the observable nominal wage structure and the reproduction floor constraint to map wages into reduction coefficients for complex labour; this answers the question of operationality.
	
	The remainder of the paper is organised as follows. Section~2 sets up the input--output model incorporating heterogeneous complex labour and the reproduction floor of labour power. Section~3 introduces the labour reproduction matrix, defines the value feasible region, and proves its geometric properties and existence condition. Section~4 introduces a uniform profit rate, constructs the price feasible region, and proves its monotonic shrinking property. Section~5 reformulates the transformation problem as an intersection problem between a hyperplane and a feasible cone, and proves the universal compatibility theorem. Section~6 proposes the mapping method and establishes a linear isomorphism between the wage space and the reduction-coefficient space. Section~7 presents a three-sector numerical example to illustrate the geometric mechanism. Section~8 carries out an empirical calibration using China's 2017 inter-provincial input--output table. Section~9 concludes.
	
	% ================================================================
	\section{Theoretical Model and Basic Setup}
	% ================================================================
	
	Consider a closed economy with $n\ge 2$ sectors, each producing a single commodity. Production technology is described by intermediate inputs, fixed capital stocks, and direct labour inputs. The consumption basket required for the reproduction of labour power may differ across sectors.
	
	\begin{definition}[Intermediate inputs, capital stock and depreciation matrix]
		Let $A=(a_{ij})\in \R_+^{n\times n}$ be the intermediate input matrix, where $a_{ij}$ denotes the amount of commodity $i$ required to produce one unit of commodity $j$.
		Let $K=(\kappa_{ij})\in \R_+^{n\times n}$ be the capital stock matrix, where $\kappa_{ij}$ denotes the stock of commodity $i$ that must be advanced to produce one unit of commodity $j$.
		Let the depreciation rate of sector $j$ be $\delta_j\in(0,1]$ and write $\hat\delta=\mathrm{diag}(\delta_1,\dots,\delta_n)$. The current depreciation matrix and the composite material input matrix are defined respectively as
		\[
		D:=K\hat\delta ,\qquad 
		\tilde A:=A+D .
		\]
	\end{definition}
	
	The above definition distinguishes ``stock advances'' from ``current transfers'': the fixed capital stock $K$ is the basis on which profit is claimed, whereas only the depreciation part $D$ enters the transfer of value to the product.
	
	\begin{definition}[Labour coefficients and consumption matrix]
		Let $L=\mathrm{diag}(l_1,\dots,l_n)$, where $l_j>0$ is the amount of direct labour of sector $j$ required to produce one unit of commodity $j$.
		Let $B=(\beta_{ij})\in \R_+^{n\times n}$ be the consumption coefficient matrix, where $\beta_{ij}$ denotes the amount of commodity $i$ that a worker of sector $j$ needs to consume per unit of labour performed. The $j$-th column of $B$ represents the minimum consumption basket that sustains the reproduction of labour power in that sector.
	\end{definition}
	
	To guarantee that the model carries realistic economic meaning, the following basic conditions are imposed on the technology and the consumption structure.
	
	\begin{assumption}[Technical conditions]\label{ass:A_L}
		The composite material input matrix $\tilde A$ satisfies:
		\begin{enumerate}[label=(\roman*)]
			\item $\tilde A\ge 0$;
			\item $\tilde A$ is irreducible;
			\item $\lambda_{\max}(\tilde A)<1$.
		\end{enumerate}
		Moreover, for all $j=1,\dots,n$, $l_j>0$.
	\end{assumption}
	
	Assumption~\ref{ass:A_L} implies that after compensating for raw material consumption and fixed capital depreciation the economy still possesses the technical possibility to reproduce itself. Irreducibility reflects the ubiquitous direct or indirect input linkages among sectors in a modern economy.
	
	\begin{assumption}[Consumption network condition]\label{ass:B}
		The consumption matrix $B\ge 0$ has no zero column; i.e.\ for each $j$ there exists at least one $i$ such that the element $\beta_{ij}>0$.
	\end{assumption}
	
	This assumption only requires that each type of worker corresponds to a non-trivial reproduction basket; it does not require that they directly consume all products. Because of the indirect linkages in the input--output network, these consumption needs are propagated throughout the whole production system via the Leontief inverse.
	
	In the case of simple reproduction with zero profit, the price system $\p=(p_1,\dots,p_n)^T>\mathbf 0$ only needs to cover the composite material costs and the wage bill:
	\begin{equation}\label{eq:price_simple}
		\p^T=\p^T\tilde A+\w^T L ,
	\end{equation}
	where $\w=(w_1,\dots,w_n)^T>\mathbf 0$ is the vector of sectoral nominal wages. At the same time, for each type of labour power to be able to reproduce itself continuously, its nominal wage must be at least sufficient to purchase its minimum consumption basket:
	\begin{equation}\label{eq:wage_constraint}
		\w^T\ge \p^T B .
	\end{equation}
	
	% ================================================================
	\section{Labour Reproduction Matrix and Value Feasible Region}
	% ================================================================
	
	From Assumption~\ref{ass:A_L}, $I-\tilde A$ is invertible, and one can solve for the price vector from \eqref{eq:price_simple}:
	\begin{equation}\label{eq:price_solution}
		\p^T=\w^T L(I-\tilde A)^{-1}.
	\end{equation}
	Substituting this into the wage-floor constraint \eqref{eq:wage_constraint} yields
	\[
	\w^T \ge \w^T L(I-\tilde A)^{-1}B .
	\]
	This naturally induces the following matrix.
	
	\begin{definition}[Basic labour reproduction matrix]\label{def:M0}
		Define
		\begin{equation}\label{eq:M0}
			M_0:=L(I-\tilde A)^{-1}B .
		\end{equation}
		The wage constraint under simple reproduction can then be written compactly as
		\begin{equation}\label{eq:wage_M0}
			\w^T(I-M_0)\ge \mathbf 0^T .
		\end{equation}
	\end{definition}
	
	Element $(i,j)$ of $M_0$ represents how much direct labour of sector $i$ is needed, through the entire network of direct and indirect inputs, to sustain one unit of labour capacity of a sector-$j$ worker. It is easy to verify that all elements of $M_0$ are dimensionless pure numbers, and their magnitudes do not depend on the units of measurement of products or on the choice of the price vector; $M_0$ reflects the labour reproduction ratios jointly determined by the technology and the consumption structure.
	
	\begin{lemma}[Properties of $M_0$]\label{lem:M_spectrum}
		Under Assumptions~\ref{ass:A_L} and \ref{ass:B}, $M_0$ is a strictly positive matrix, hence irreducible and aperiodic. Its maximal eigenvalue $\lambda^*:=\lambda_{\max}(M_0)$ is a simple positive real root, and there correspond unique (up to scalar multiples) strictly positive left and right Perron eigenvectors.
	\end{lemma}
	
	Lemma~\ref{lem:M_spectrum} indicates that in a modern economy with a highly developed division of labour, the reproduction of any type of labour is linked to other labour types through the input chain. The Perron eigenvector of $M_0$ can therefore be understood as a ``benchmark direction'' endogenously determined by the objective reproduction network.
	
	Now introduce complex labour. Let $\cvec=(c_1,\dots,c_n)^T>\mathbf 0$ be the vector of reduction coefficients that convert complex labour into simple labour, where $c_j$ indicates how many units of ``simple labour'' one unit of direct labour of sector $j$ is equivalent to. Given $\cvec$, the vector of labour values of commodities $\mathbf v=(v_1,\dots,v_n)^T$ is determined by
	\begin{equation}\label{eq:value_equation}
		\mathbf v^T=\mathbf v^T\tilde A+\cvec^T L .
	\end{equation}
	Since $I-\tilde A$ is invertible,
	\begin{equation}\label{eq:value_solution}
		\mathbf v^T=\cvec^T L(I-\tilde A)^{-1}.
	\end{equation}
	The value of labour power in sector $j$, i.e.\ the labour value of the consumption basket required to sustain one unit of its labour capacity, is
	\[
	\sigma_j = \mathbf v^T\mathbf b_j = [\cvec^T M_0]_j .
	\]
	Hence, for sector $j$ the difference between the new value created per unit of labour and the value of its labour power defines the sector's objective rate of exploitation:
	\[
	e_j(\cvec)=\frac{c_j-[\cvec^T M_0]_j}{[\cvec^T M_0]_j}.
	\]
	For social reproduction to be sustainable, the new value created by each type of worker must be at least as large as the value of its labour power, i.e.\
	\[
	c_j\ge [\cvec^T M_0]_j,\qquad j=1,\dots,n,
	\]
	which can be written in vector form as
	\[
	\cvec^T(I-M_0)\ge \mathbf 0^T .
	\]
	
	\begin{definition}[Value feasible cone and normalised value feasible region]\label{def:value_cone}
		Given $(\tilde A,L,B)$, define the value feasible cone
		\begin{equation}\label{eq:value_cone}
			\widetilde{\Theta}^{\mathrm{val}}:=\left\{\cvec\in\R_+^n\setminus\{\mathbf 0\}: \cvec^T(I-M_0)\ge \mathbf 0^T\right\}.
		\end{equation}
		Its strict interior (where all inequalities are strict) is denoted by $\widetilde{\Theta}^{\circ}$.
		Taking sector~1's labour as the num\'eraire, the normalised section (value feasible region) on the hyperplane $c_1=1$ is defined as
		\begin{equation}\label{eq:value_slice}
			\Theta^{\mathrm{val}}:=\widetilde{\Theta}^{\mathrm{val}}\cap\{\cvec\in\R_+^n: c_1=1\},
			\qquad
			\Theta^{\circ}:=\widetilde{\Theta}^{\circ}\cap\{\cvec\in\R_+^n: c_1=1\}.
		\end{equation}
	\end{definition}
	
	$\widetilde{\Theta}^{\mathrm{val}}$ contains all the structures of complex labour reduction that do not violate the reproduction floor of labour power. If a reduction vector falls outside the cone, then there exists at least one sector whose workers create less new value than the value of their labour power, making the corresponding reproduction structure unsustainable. The normalised section simply removes one degree of freedom and facilitates the visualisation of three-sector numerical examples.
	
	Whether the value feasible region exists depends on whether the economy still yields a surplus after compensating for composite material inputs and labour-power reproduction. This condition can be judged strictly by the spectral radius; it is a corollary of the Fundamental Marxian Theorem.
	
	\begin{theorem}[Existence of the value feasible region]\label{thm:existence}
		Under Assumptions~\ref{ass:A_L} and~\ref{ass:B}, the following statements are equivalent:
		\begin{enumerate}[label=(\roman*)]
			\item the strictly value feasible cone is non-empty, i.e.\ $\widetilde{\Theta}^{\circ}\neq\emptyset$;
			\item $\lambda_{\max}(M_0)<1$;
			\item $\lambda_{\max}(\tilde A+BL)<1$.
		\end{enumerate}
	\end{theorem}
	
	Theorem~\ref{thm:existence} shows that the feasibility of reducing complex labour is not a matter of subjective judgement, but is objectively determined by the observable physical reproduction network: as long as society can still leave a surplus after covering raw materials, depreciation, and workers' subsistence, a strictly positive value feasible region necessarily exists; conversely, if the surplus condition fails, no reduction structure that attempts to sustain a positive surplus value can be viable.
	
	When the condition of Theorem~\ref{thm:existence} holds, the value feasible region is not only non-empty but also possesses a well-behaved geometric structure.
	
	\begin{proposition}[The value feasible region is a compact convex set]\label{prop:convex}
		If $\lambda_{\max}(M_0)<1$, then $\Theta^{\mathrm{val}}$ is a non-empty bounded closed convex set on the hyperplane $c_1=1$. Concretely, for each $j=2,\dots,n$ there exist finite constants
		\[
		0<\underline c_j\le \overline c_j<\infty
		\]
		such that for all $\cvec\in\Theta^{\mathrm{val}}$, $\underline c_j\le c_j\le \overline c_j$.
	\end{proposition}
	
	Proposition~\ref{prop:convex} tells us that, after choosing a benchmark labour, the relative reduction coefficient of any sector's labour can neither be blown up to infinity nor compressed to zero. Of course society's evaluation of the complexity of different types of labour can vary within a certain range, but that range is always strictly bounded by the objective reproduction structure. In other words, the reduction coefficients, treated as a single unknown in the traditional transformation debate, become here a set with clearly defined boundaries.
	
	Inside the value feasible region there is a special location uniquely determined by the reproduction network.
	
	\begin{theorem}[Equal exploitation point]\label{thm:equilibrium}
		Let $\lambda^*=\lambda_{\max}(M_0)<1$ and let $\mathbf y^*>\mathbf 0$ be the normalised left Perron eigenvector of $M_0$, satisfying
		\[
		\mathbf y^{*T}M_0=\lambda^*\mathbf y^{*T},\qquad y_1^*=1.
		\]
		Then:
		\begin{enumerate}[label=(\roman*)]
			\item $\mathbf y^*\in\Theta^\circ$;
			\item at $\mathbf y^*$ the rates of exploitation are exactly equal across all sectors;
			\item within $\Theta^{\mathrm{val}}$, $\mathbf y^*$ is the unique point at which the exploitation rates are equalised;
			\item the common exploitation rate is $e^*=\dfrac{1-\lambda^*}{\lambda^*}>0$.
		\end{enumerate}
	\end{theorem}
	
	Theorem~\ref{thm:equilibrium} generalises Morishima's transformation model. The equal-exploitation point is not an equilibrium point that is necessarily realised by market mechanisms, nor should it be understood as an empirical approximation of the actual wage structure; it merely points out that among all the structures of complex labour reduction compatible with the reproduction floor of labour power, there exists a unique point that fully equalises the objective rates of exploitation across sectors.
	
	% ================================================================
	\section{Uniform Profit Rate and Price Feasible Region}
	% ================================================================
	
	The previous section discussed the objective boundaries of the value feasible region within the framework of simple reproduction with zero profit. However, real-world price systems contain profit. For simplicity and in keeping with the research tradition of the transformation problem, this section introduces a uniform profit rate and examines how the feasible range of wage structures is additionally constrained by the profit requirement.
	
	Let the uniform profit rate be $r\ge 0$. Following Marx's setting, profit is claimed on the advanced fixed capital stock $K$. The price equation then becomes, instead of \eqref{eq:price_simple},
	\begin{equation}\label{eq:price_profit}
		\p^T = \p^T\tilde A + \w^T L + r\p^T K .
	\end{equation}
	As long as the matrix $I-\tilde A-rK$ is invertible, the price vector can be expressed as a linear mapping of the wage vector:
	\begin{equation}\label{eq:price_profit_solution}
		\p^T = \w^T L(I-\tilde A-rK)^{-1}.
	\end{equation}
	Pure technical conditions determine an upper bound $r_A$ for the profit rate, which satisfies
	\[
	\lambda_{\max}(\tilde A + r_A K) = 1 .
	\]
	For $r\in[0,r_A)$ the matrix $I-\tilde A-rK$ remains invertible, so that the price system is uniquely determined once the wage structure is given.
	
	Substituting \eqref{eq:price_profit_solution} into the wage-floor constraint $\w^T\ge \p^T B$ gives
	\[
	\w^T \ge \w^T L(I-\tilde A-rK)^{-1}B .
	\]
	This naturally leads to a labour reproduction matrix parameterised by the profit rate.
	
	\begin{definition}[Parameterised labour reproduction matrix]\label{def:M_r}
		For any $r\in[0,r_A)$, define
		\begin{equation}\label{eq:M_r}
			M(r) := L(I-\tilde A-rK)^{-1}B .
		\end{equation}
		The wage reproduction constraint can then be written as
		\begin{equation}\label{eq:wage_Mr}
			\w^T(I-M(r)) \ge \mathbf 0^T .
		\end{equation}
	\end{definition}
	
	Completely analogous to $M_0$, the element of $M(r)$ represents, under a price system with uniform profit rate $r$, the amount of direct labour of each sector that the entire economic network must expend in order to sustain one unit of labour capacity of a sector-$j$ worker. As $r$ rises, the nominal wage payment required to sustain the same consumption basket systematically increases.
	
	\begin{proposition}[Monotonicity and spectral properties of $M(r)$]\label{prop:M_r}
		Under Assumptions~\ref{ass:A_L}, \ref{ass:B} and $K\neq 0$, for any $r\in[0,r_A)$:
		\begin{enumerate}[label=(\roman*)]
			\item $M(r)$ is a strictly positive matrix and is continuous in $r$;
			\item every element of $M(r)$ is strictly increasing in $r$;
			\item the spectral radius $\lambda_{\max}(M(r))$ is strictly increasing in $r$, and as $r\uparrow r_A$, $\lambda_{\max}(M(r))\to +\infty$.
		\end{enumerate}
	\end{proposition}
	
	Proposition~\ref{prop:M_r} reveals the positive relationship between the rising profit rate and the reproduction pressure on wages: the higher the profit claimed by capital, the larger the nominal wage required for workers to maintain the same standard of living, and accordingly the narrower the feasible space for the wage structure. The technically maximal profit rate $r_A$ does not mean the economy can actually operate at that upper bound, because before that point the reproduction floor of labour power may already become impossible to satisfy.
	
	\begin{theorem}[Maximal feasible profit rate]\label{thm:r_star}
		If $\lambda_{\max}(M_0)<1$, there exists a unique critical profit rate $r^*\in(0,r_A)$ such that
		\[
		\lambda_{\max}(M(r^*)) = 1 .
		\]
		Moreover:
		\begin{enumerate}[label=(\roman*)]
			\item when $0\le r<r^*$, $\lambda_{\max}(M(r))<1$ and the wage reproduction constraint admits a strictly feasible solution;
			\item when $r=r^*$, the constraint is exactly at the critical boundary;
			\item when $r>r^*$, no feasible solution exists.
		\end{enumerate}
	\end{theorem}
	
	Thus $r^*$ is the true profit-rate ceiling compatible with workers' reproduction floor. The gap between $r^*$ and $r_A$ captures the hard constraint imposed by the conditions of labour-power reproduction on the drive for capital valorisation.
	
	Similar to the value feasible cone defined earlier, we first define the primitive feasible cone of wage vectors:
	\[
	\widetilde{\Theta}^{\mathrm{price}}(r) := \left\{ \w\in\R_+^n\setminus\{\mathbf 0\} : \w^T(I-M(r))\ge \mathbf 0^T \right\}.
	\]
	To remove the arbitrariness of the absolute scale of wages, we take the normalised section $w_1=1$ and obtain the feasible region of relative nominal wages.
	
	\begin{definition}[Price feasible region]\label{def:Theta_price}
		For any $r\in[0,r^*]$, define
		\begin{equation}\label{eq:Theta_price}
			\Theta^{\mathrm{price}}(r) := \widetilde{\Theta}^{\mathrm{price}}(r) \cap \left\{ \w\in\R_+^n : w_1=1 \right\}.
		\end{equation}
	\end{definition}
	
	It is important to stress that although $\Theta^{\mathrm{price}}(r)$ and the value feasible region $\Theta^{\mathrm{val}}$ both lie in the space of normalised positive vectors, their economic meanings are completely different: the former consists of relative nominal wage structures, whereas the latter consists of reduction coefficients of complex labour. The geometric comparisons that follow only concern the inclusion and intersection relations between them as sets of positive vectors, and do not identify the two as the same thing.
	
	The compression of the bargaining space of wages caused by a rising profit rate is expressed geometrically as a monotonic shrinking of the price feasible region.
	
	\begin{theorem}[Monotonic shrinking of the price feasible region]\label{thm:duality}
		Under Assumptions~\ref{ass:A_L}, \ref{ass:B} and $K\neq 0$, if $\lambda_{\max}(M_0)<1$, then for any $0\le r_1<r_2\le r^*$,
		\[
		\Theta^{\mathrm{price}}(r_2)\subsetneq \Theta^{\mathrm{price}}(r_1).
		\]
		In particular, when $r=r^*$, $\Theta^{\mathrm{price}}(r^*)$ degenerates to a single point, namely the normalised positive left Perron eigenvector of $M(r^*)$.
	\end{theorem}
	
	Theorem~\ref{thm:duality} illustrates vividly that the consequence of capital pursuing a higher profit rate is not only a fall in the wage share, but also a narrowing of the range of relative wage structures that society can sustain. As the profit rate approaches $r^*$, the room for negotiation over the wage structure is squeezed to the limit; any deviation from that unique feasible direction would make the reproduction of some workers' labour power unsustainable.
	
	% ================================================================
	\section{Geometric Solution to the Transformation Problem}
	% ================================================================
	
	The preceding two sections characterised the structural feasible regions from the value space and from the price space separately. This section brings together the constraints from these two dimensions, reformulates the transformation problem in geometric language, and provides a necessary and sufficient condition for the two macro aggregate equalities to hold simultaneously.
	
	Given the vector of social gross output $\mathbf x>\mathbf 0$, for any reduction coefficient vector $\cvec\in \widetilde{\Theta}^{\mathrm{val}}$ define total value and total surplus value:
	\begin{align}
		F(\cvec) &:= \cvec^T L(I-\tilde A)^{-1}\mathbf x , \label{eq:value_functional}\\
		S(\cvec) &:= \cvec^T L\mathbf x - \cvec^T M_0 L\mathbf x
		= \cvec^T L(I-\tilde A)^{-1}(I-\tilde A-BL)\mathbf x . \label{eq:total_surplus}
	\end{align}
	
	In the price space, given a uniform profit rate $r\in[0,r^*]$ and some relative wage structure $\w^{rel}\in\Theta^{\mathrm{price}}(r)$, the corresponding total price and total profit are
	\begin{align}
		P(r,\w^{rel}) &:= \w^{rel,T} L(I-\tilde A-rK)^{-1}\mathbf x , \label{eq:P_star}\\
		\Pi(r,\w^{rel}) &:= r\,\w^{rel,T} L(I-\tilde A-rK)^{-1}K\mathbf x . \label{eq:Pi_star}
	\end{align}
	One can then define the macro profit share corresponding to this price system,
	\begin{equation}\label{eq:gamma_rw}
		\gamma(r,\w^{rel}) := \frac{\Pi(r,\w^{rel})}{P(r,\w^{rel})} \in (0,1).
	\end{equation}
	
	Within the framework of feasible regions, the transformation problem demands finding a reduction vector $\cvec$ such that
	\[
	F(\cvec) = P(r,\w^{rel}), \qquad S(\cvec) = \Pi(r,\w^{rel}).
	\]
	These two conditions can be decomposed into a condition that determines only direction and one that determines only scale.
	
	\begin{theorem}[Direction and scale]\label{thm:two_equalities}
		Fix $r\in[0,r^*]$ and $\w^{rel}\in\Theta^{\mathrm{price}}(r)$, and denote
		\[
		P^*:=P(r,\w^{rel}),\quad \Pi^*:=\Pi(r,\w^{rel}),\quad \gamma:=\frac{\Pi^*}{P^*}.
		\]
		Then for any $\cvec\in\widetilde{\Theta}^{\mathrm{val}}$, the two macro aggregate equalities $F(\cvec)=P^*$ and $S(\cvec)=\Pi^*$ hold simultaneously if and only if $\cvec$ satisfies both of the following:
		\begin{align}
			\cvec^T \boldsymbol{\eta}(\gamma) &= 0 , \label{eq:second_hyperplane}\\
			\cvec^T \boldsymbol{\xi} &= P^* , \label{eq:first_hyperplane}
		\end{align}
		where
		\begin{align}
			\boldsymbol{\eta}(\gamma) &:= L(I-\tilde A)^{-1}\Big[(\tilde A+BL)\mathbf x-(1-\gamma)\mathbf x\Big], \label{eq:eta_gamma}\\
			\boldsymbol{\xi} &:= L(I-\tilde A)^{-1}\mathbf x . \label{eq:xi_vector}
		\end{align}
		Condition \eqref{eq:second_hyperplane} only determines the direction of $\cvec$ in the cone (it is homogeneous), while condition \eqref{eq:first_hyperplane} uniquely determines its scale once the direction is given.
	\end{theorem}
	
	Theorem~\ref{thm:two_equalities} translates the transformation problem into geometric language: the homogeneous hyperplane corresponding to the second equality,
	\[
	H_2(\gamma) := \left\{ \cvec\in\R^n : \cvec^T\boldsymbol{\eta}(\gamma)=0 \right\},
	\]
	selects those reduction directions that make total surplus value proportional to total profit with the correct factor; the inhomogeneous hyperplane corresponding to the first equality,
	\[
	H_1(P^*) := \left\{ \cvec\in\R^n : \cvec^T\boldsymbol{\xi}=P^* \right\},
	\]
	then picks, among those directions, the point at which total value equals total price.
	
	Thus the transformation problem is converted to: given the desired $(P^*,\gamma)$, does $H_2(\gamma)$ have a common intersection line with the value feasible cone $\widetilde{\Theta}^{\mathrm{val}}$, and moreover, does there exist a point on that intersection line that also lies on $H_1(P^*)$? The following proposition shows that once the direction problem is solved, the scale problem can always be resolved automatically by a single proportional rescaling.
	
	\begin{proposition}[Direction problem]\label{prop:ray_scale}
		If there exists a normalised direction $\bar{\cvec}\in H_2(\gamma)\cap\Theta^{\mathrm{val}}$ (or $\Theta^{\circ}$), then let
		\[
		\cvec = \frac{P^*}{F(\bar{\cvec})}\,\bar{\cvec},
		\]
		and we obtain $\cvec\in\widetilde{\Theta}^{\mathrm{val}}$ (or $\widetilde{\Theta}^{\circ}$) which simultaneously satisfies the two macro aggregate equalities. Conversely, any $\cvec\in\widetilde{\Theta}^{\mathrm{val}}$ that satisfies the two equalities, after normalisation, must fall into $H_2(\gamma)\cap\Theta^{\mathrm{val}}$.
	\end{proposition}
	
	At this point the problem is entirely focused on whether the direction hyperplane $H_2(\gamma)$ intersects the normalised value feasible region $\Theta^{\mathrm{val}}$. To provide a criterion for this intersection, introduce the extended reproduction operator
	\[
	\hat A := \tilde A+BL,\qquad \hat S := L(I-\hat A)^{-1},
	\]
	and the sectoral critical profit shares endogenously determined by the physical reproduction network:
	\[
	\gamma_i^{\mathrm{crit}} := 1 - \frac{[\hat S\hat A\mathbf x]_i}{[\hat S\mathbf x]_i},\qquad i=1,\dots,n.
	\]
	
	\begin{theorem}[Compatibility interval]\label{thm:universal_compatibility}
		Suppose the physical surplus condition holds, i.e.\ $\lambda_{\max}(\tilde A+BL)<1$. Fix $r\in[0,r^*]$ and $\w^{rel}\in\Theta^{\mathrm{price}}(r)$, and let $\gamma=\gamma(r,\w^{rel})$. Then:
		\begin{enumerate}[label=(\roman*)]
			\item If $\gamma\in(\gamma_{\min}^{\mathrm{crit}},\gamma_{\max}^{\mathrm{crit}})$, there exists $\cvec\in\widetilde{\Theta}^{\circ}$ satisfying the two macro aggregate equalities simultaneously;
			\item If $\gamma$ equals $\gamma_{\min}^{\mathrm{crit}}$ or $\gamma_{\max}^{\mathrm{crit}}$, there exists $\cvec\in\widetilde{\Theta}^{\mathrm{val}}$ satisfying the two equalities (boundary solution);
			\item If $\gamma\notin[\gamma_{\min}^{\mathrm{crit}},\gamma_{\max}^{\mathrm{crit}}]$, no $\cvec\in\widetilde{\Theta}^{\mathrm{val}}$ can satisfy the two equalities simultaneously.
		\end{enumerate}
		Here $\gamma_{\min}^{\mathrm{crit}}=\min_i\gamma_i^{\mathrm{crit}}$ and $\gamma_{\max}^{\mathrm{crit}}=\max_i\gamma_i^{\mathrm{crit}}$.
	\end{theorem}
	
	Theorem~\ref{thm:universal_compatibility} shows that as long as the macro profit share falls within the compatibility interval endogenously determined by the objective reproduction network, there necessarily exists a set of strictly positive reduction coefficients of complex labour that makes the labour value system and the production price system consistent at the macro level.
	
	However, the above existence result leaves an open question: in a real economy, starting from the observed data, how should we systematically find the set of reduction coefficients that is ``selected'' by the social process? Theorem~\ref{thm:universal_compatibility} tells us that such coefficients do exist, but it does not provide any operational method for finding them. The next section answers this question directly.
	
	% ================================================================
	\section{Linear Transformation from Wages to Reduction Coefficients}
	% ================================================================
	
	The previous sections have proved that within the value feasible region $\Theta^{\mathrm{val}}$ and the price feasible region $\Theta^{\mathrm{price}}(r)$ there each exists a set of normalised vectors satisfying the reproduction constraints. A natural question arises: is there a systematic correspondence between these two feasible regions? If one can establish a mapping from the observable nominal wage to the unobservable reduction coefficients, then the reduction of complex labour is no longer a purely speculative problem, but can enter the empirical domain.
	
	The two feasible regions are defined respectively by $\cvec^T(I-M_0)\ge\mathbf 0^T$ and $\w^T(I-M(r))\ge\mathbf 0^T$. Their constraint structures are highly symmetric: $M_0$ and $M(r)$ are both positive matrices based on the same technology and consumption network, the only difference being that the latter includes the price markup of the profit rate. This symmetry suggests that there may exist a correspondence between vectors in the two spaces, mediated by the reproduction constraints themselves. From the numerical example in Section~7 one can observe that the price feasible region and the value feasible region have similar shapes, and a linear mapping may exist that sends the price feasible region into the value feasible region.
	
	\begin{definition}[Proportional sectoral surpluses]\label{def:proportional}
		Given a profit rate $r$, the sectoral surplus vectors of a wage vector $\w$ and a reduction coefficient vector $\cvec$ satisfy
		\begin{equation}\label{eq:proportional}
			\w^T(I-M(r)) \;\propto\; \cvec^T(I-M_0),
		\end{equation}
		i.e.\ the two are proportional. In symbols, there exists some positive number $\theta>0$ such that $\w^T(I-M(r)) = \theta\,\cvec^T(I-M_0)$. As will be seen below, $\theta$ cancels out automatically during normalisation.
	\end{definition}
	
	It must be emphasised that Assumption~\ref{def:proportional} is not a theorem deduced from some more fundamental economic behavioural hypothesis (such as profit-rate equalisation or exploitation-rate equalisation); rather, it is the condition we introduce in order to identify the reduction coefficients from the observed wages. Besides the geometric intuition, there is also an economic intuition behind this assumption: for sectors whose reproduction constraint is ``tight'' (where wages barely suffice for the reproduction of labour power), the surplus in the two spaces should be consistent. For instance, in a sector where workers barely make ends meet, wages can only buy essential goods, with no monetary surplus left; then all of the workers' labour should also just be enough to exchange for the necessary means of subsistence, leaving no surplus labour.
	
	Under Assumption~\ref{def:proportional}, there exists $\theta>0$ such that
	\begin{equation}\label{eq:theta_relation}
		\w^T(I-M(r)) = \theta\,\cvec^T(I-M_0).
	\end{equation}
	Under the conditions $\lambda_{\max}(M_0)<1$ and $r\in[0,r^*)$, the matrix $I-M_0$ is invertible. Right-multiplying both sides of \eqref{eq:theta_relation} by $(I-M_0)^{-1}$ gives
	\begin{equation}\label{eq:derivation_1}
		\cvec^T = \frac{1}{\theta}\,\w^T(I-M(r))(I-M_0)^{-1}.
	\end{equation}
	Equation \eqref{eq:derivation_1} shows that given $\w$ and $\theta$, $\cvec$ is uniquely determined. However, what we actually care about is not the absolute scale of $\cvec$---the absolute level of the reduction coefficients can be calibrated ex post by the first macro aggregate equality (total price equals total value)---but rather its \textbf{relative structure} across sectors, i.e.\ the direction of $\cvec$. When we normalise $\cvec$ (for instance by setting $c_1=1$), the scalar factor $1/\theta$ appears simultaneously in both the numerator and the denominator and cancels out. Concretely, let $\cvec_{\mathrm{raw}}^T = \w^T(I-M(r))(I-M_0)^{-1}$, then $\cvec^T = (1/\theta)\,\cvec_{\mathrm{raw}}^T$, and after normalisation
	\[
	\frac{\cvec}{c_1} = \frac{(1/\theta)\,\cvec_{\mathrm{raw}}}{(1/\theta)\,(\cvec_{\mathrm{raw}})_1} = \frac{\cvec_{\mathrm{raw}}}{(\cvec_{\mathrm{raw}})_1},
	\]
	completely independent of $\theta$. Hence we do not need to solve for the specific value of $\theta$.
	
	Based on this observation, we define the following mapping.
	
	\begin{definition}[Homeomorphic mapping]\label{def:dual_map}
		Under the conditions $\lambda_{\max}(M_0)<1$ and $r\in[0,r^*)$, the matrices $I-M_0$ and $I-M(r)$ are both invertible. Define the linear transformation
		\begin{equation}\label{eq:Phi}
			\Phi_r(\w)^\top = \w^\top (I-M(r))(I-M_0)^{-1}.
		\end{equation}
		The induced mapping on the normalised section is denoted by
		\begin{equation}\label{eq:phi}
			\phi_r(\w) := \frac{\Phi_r(\w)}{\Phi_r(\w)_1}.
		\end{equation}
	\end{definition}
	
	The invertibility of $(I-M(r))(I-M_0)^{-1}$ directly implies the following properties.
	
	\begin{theorem}[Properties of the homeomorphic mapping]\label{thm:mapping_properties}
		For any $r\in[0,r^*)$:
		\begin{enumerate}[label=(\roman*)]
			\item $\Phi_r$ is a linear isomorphism from $\widetilde{\Theta}^{\mathrm{price}}(r)$ to $\widetilde{\Theta}^{\mathrm{val}}$, and $\phi_r$ is a homeomorphism from $\Theta^{\mathrm{price}}(r)$ to $\Theta^{\mathrm{val}}$.
			\item The homeomorphism implies \textbf{boundary preservation}: $\w\in\Theta^{\mathrm{price}}(r)$ makes the $j$-th reproduction constraint an equality (i.e.\ $[\w^T(I-M(r))]_j=0$) if and only if $\phi_r(\w)$ also makes the $j$-th value constraint an equality. In particular, the vertices of $\Theta^{\mathrm{price}}(r)$ are mapped to the vertices of $\Theta^{\mathrm{val}}$.
		\end{enumerate}
	\end{theorem}
	
	This theorem indicates that, starting from any wage structure that satisfies the reproduction floor, we can uniquely find a set of reduction coefficients that also satisfy the reproduction floor, and the ``tightness'' structure of the wage constraints is preserved.
	As the profit rate tends to zero, the matrix $M(r)$ degenerates to $M_0$. In this case, if the structures of technical inputs and labour compensation shares remain unchanged, the reduction coefficients given by the mapping method correspond exactly to the ``relative wages at zero profit rate'', i.e.\ the relative evaluation that different types of labour ought to receive in the social reproduction network after stripping away the distortion of the nominal price system caused by the profit requirement of capital. In this sense, the mapping method essentially ``de-profit-rates'': it uses the objective reproduction network to filter out the impact of profit on the wage structure and extracts the labour complexity implicit in the wage information---complexity that has been evaluated by the market.
	
	Given an observed profit rate $r^{\mathrm{obs}}$ and an observed relative wage $\w^{\mathrm{obs}}$, the reduction coefficients are constructed by the following steps:
	
	\begin{enumerate}
		\item Compute the normalised reduction coefficients $\cvec^{\mathrm{mapping}} = \phi_{r^{\mathrm{obs}}}(\w^{\mathrm{obs}})$. This step essentially solves the linear system
		\[
		\cvec^T(I-M_0) = \w^{\mathrm{obs},T}(I-M(r^{\mathrm{obs}})),
		\]
		normalises its solution $\cvec_{\mathrm{raw}}$ to obtain $\cvec^{\mathrm{mapping}} = \cvec_{\mathrm{raw}} / (\cvec_{\mathrm{raw}})_1$.
		
		\item Calibrate the scale using the first macro aggregate equality ``total price equals total value'' and compute the monetary expression coefficient
		\[
		\kappa = \frac{P^*}{F(\cvec^{\mathrm{mapping}})},
		\]
		where $P^*$ is the observed total price and $F(\cdot)$ is the total value function defined in \eqref{eq:value_functional}. The final reduction coefficients are $\tilde{\cvec} = \kappa \cdot \cvec^{\mathrm{mapping}}$. At this point the first equality holds automatically.
		
		\item The deviation of the second equality ``total profit equals total surplus value'' is computed ex post in this procedure as
		\[
		\Delta_2 = \frac{|S(\tilde{\cvec}) - \Pi^*|}{\Pi^*}.
		\]
		This deviation reflects the discrepancy between the assumption of proportional sectoral surpluses and the actual economic structure.
	\end{enumerate}
	
	% ================================================================
	\section{Numerical Example}
	% ================================================================
	
	This section uses a three-sector numerical example to fully demonstrate the working mechanism of the theoretical framework developed above, giving the abstract geometric arguments a concrete quantitative illustration. All parameters in the example satisfy the assumptions made earlier.
	
	The economy contains three sectors. The intermediate input matrix and the fixed capital stock matrix are set as
	\[
	A=
	\begin{bmatrix}
		0.15 & 0.18 & 0.12\\
		0.20 & 0.12 & 0.15\\
		0.10 & 0.15 & 0.18
	\end{bmatrix},
	\qquad
	K=
	\begin{bmatrix}
		0.40 & 0.35 & 0.30\\
		0.30 & 0.45 & 0.25\\
		0.25 & 0.30 & 0.40
	\end{bmatrix}.
	\]
	
	The vector of sectoral depreciation rates is $\boldsymbol{\delta}=(0.10,\,0.12,\,0.08)^T$, and $\hat\delta=\mathrm{diag}(\boldsymbol{\delta})$. The current depreciation matrix is $D=K\hat\delta$, and the composite material input matrix is
	\[
	\tilde A = A+D =
	\begin{bmatrix}
		0.190 & 0.222 & 0.144\\
		0.230 & 0.174 & 0.170\\
		0.125 & 0.186 & 0.212
	\end{bmatrix}.
	\]
	One can verify that $\lambda_{\max}(\tilde A)=0.5519<1$, satisfying Assumption~\ref{ass:A_L}.
	
	The direct labour coefficient matrix, consumption matrix, and gross output vector are set as
	\[
	L=\mathrm{diag}(0.40,\,0.60,\,0.35),\qquad
	B=
	\begin{bmatrix}
		0.25 & 0.30 & 0.20\\
		0.20 & 0.25 & 0.15\\
		0.22 & 0.28 & 0.25
	\end{bmatrix},\qquad
	\mathbf x=
	\begin{bmatrix}
		100\\
		80\\
		120
	\end{bmatrix}.
	\]
	
	From Definition~\ref{def:M0} we compute the basic labour reproduction matrix $M_0=L(I-\tilde A)^{-1}B$; its maximal eigenvalue is
	\[
	\lambda_{\max}(M_0)=0.7184<1.
	\]
	We also compute the extended material matrix $\hat A=\tilde A+BL$, obtaining $\lambda_{\max}(\hat A)=0.8749<1$. By Theorem~\ref{thm:existence}, the strictly value feasible cone is non-empty, and the normalised value feasible region $\Theta^{\mathrm{val}}$ is a bounded closed convex set.
	
	Taking sector~1's labour as the num\'eraire ($c_1=1$), the equal-exploitation point is given by the normalised left Perron eigenvector of $M_0$:
	\[
	\mathbf y^* = (1.0000,\;1.4523,\;0.8736)^T \in \Theta^{\circ}.
	\]
	The objective rate of exploitation common to all sectors at this benchmark point is
	\[
	e^* = \frac{1-\lambda_{\max}(M_0)}{\lambda_{\max}(M_0)} = 39.20\%.
	\]
	
	We choose a test profit rate $r = 11.55\%$ (which lies within $[0,r^*]$) and take the Perron eigenvector of $M(r)$ as the relative wage structure $\mathbf{w}^{\mathrm{rel}}\in\Theta^{\mathrm{price}}(r)$. The macro profit share corresponding to this wage structure is $\gamma = 11.71\%$, which falls entirely inside the compatibility interval $[10.69\%,\,15.12\%]$. Hence the hyperplane of the second equality does intersect the value feasible region.
	Figure~\ref{fig:geometry} fully displays this geometric solution process.
	
	\begin{figure}[htbp]
		\centering
		\includegraphics[width=0.85\textwidth]{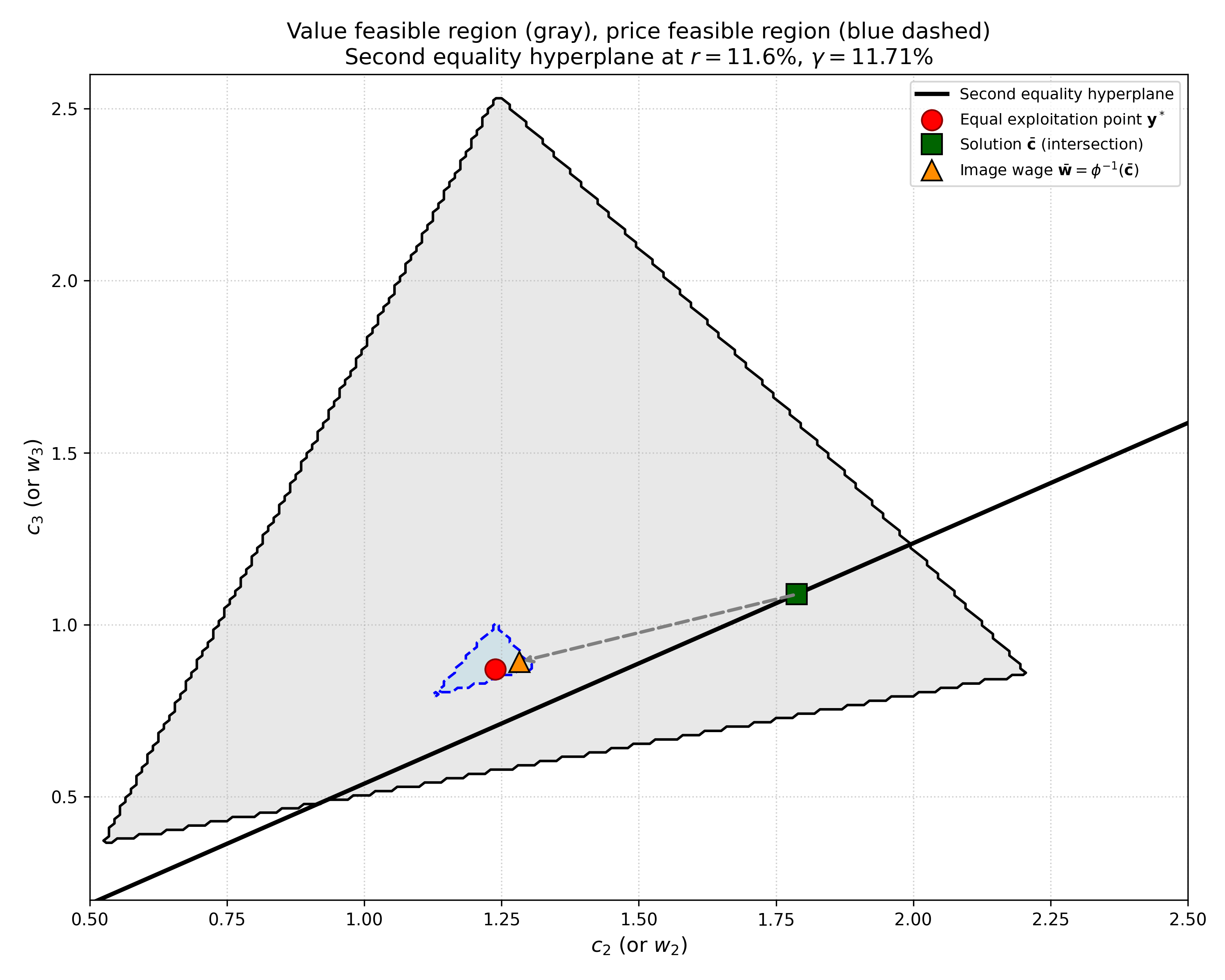}
		\caption{Value feasible region, price feasible region and geometric mapping of the two equalities (cross-section $c_1=1$)}
		\label{fig:geometry}
	\end{figure}
	
	In the figure, the grey area is the normalised value feasible region $\Theta^{\mathrm{val}}$, i.e.\ the set of all complex-labour reduction coefficient vectors that can sustain the reproduction floor of labour power in every sector. The area outlined by the blue dashed line is the price feasible region $\Theta^{\mathrm{price}}(r)$ on the same normalised section, i.e.\ the set of all relative nominal wage structures that satisfy the wage reproduction constraint. Both feasible regions are bounded closed convex sets, and the price feasible region is completely contained inside the value feasible region.
	
	The black solid line cutting through the feasible region is the direction hyperplane $H_2(\gamma)$ corresponding to the second equality. The economic meaning of this hyperplane is: any reduction coefficient vector lying on this hyperplane makes the equality ``total profit equals total surplus value'' hold. The intersection of the hyperplane with the value feasible region forms all feasible reduction directions that satisfy the second equality. The dashed line parallel to the hyperplane represents a hyperplane under an incompatible scenario---when the profit share falls outside the compatibility interval, the hyperplane would lie completely outside the feasible region, and no feasible solution can be found.
	
	The red dot $\mathbf{y}^*$ is the equal-exploitation benchmark point. It is the unique position within the value feasible region that fully equalises the objective rates of exploitation across sectors, and is given by the Perron eigenvector of the labour reproduction matrix $M_0$. This point provides a normalised comparison base that does not depend on the nominal price system.
	
	The green square $\bar{\mathbf{c}}$ is the intersection point of the hyperplane with the value feasible region, i.e.\ the feasible reduction direction that satisfies the second equality. Its coordinates are $\bar{\mathbf{c}} = (1.0000,\;1.7868,\;1.0902)^T$. On the normalised section it corresponds to a ray direction in the value feasible cone---any reduction coefficient vector proportional to this direction can satisfy the second equality. The figure shows that $\bar{\mathbf{c}}$ does not coincide with $\mathbf{y}^*$: the equal-exploitation benchmark lies on the other side of the hyperplane, indicating that the requirement of equalised exploitation rates and the holding of the macro aggregate equalities exhibit a certain distance---yet both lie inside the same value feasible region, and this distance itself is bounded by the objective structure of the reproduction network.
	
	The orange triangle $\bar{\mathbf{w}}$ is the nominal wage vector obtained by applying the inverse of the dual mapping $\phi_r^{-1}$ to $\bar{\mathbf{c}}$; its coordinates are $\bar{\mathbf{w}} = (1.0000,\;1.6874,\;1.0233)^T$. The grey dashed arrow from $\bar{\mathbf{c}}$ to $\bar{\mathbf{w}}$ intuitively shows the one-to-one correspondence of the homeomorphic mapping between the boundaries of the feasible regions. $\bar{\mathbf{w}}$ is precisely that wage structure in the price feasible region which is mapped to $\bar{\mathbf{c}}$, satisfying the wage reproduction constraint and forming a one-to-one correspondence with the solution of the second equality through the mapping method.
	
	The intersection point $\bar{\mathbf{c}}$ determines only the direction of the reduction coefficients. To satisfy the first equality ``total price equals total value'' simultaneously, the scale must be adjusted. Compute the value of the total value function in that direction:
	\[
	F(\bar{\mathbf{c}}) = \bar{\mathbf{c}}^T L(I-\tilde A)^{-1}\mathbf{x} = 401.5912.
	\]
	Under the same price system, total price is $P^* = 427.8619$. Hence the monetary expression coefficient is $\kappa = P^* / F(\bar{\mathbf{c}}) = 1.0654$. After scaling along the direction $\bar{\mathbf{c}}$, the actual solution vector on the value feasible cone is
	\[
	\tilde{\mathbf{c}} = \kappa \cdot \bar{\mathbf{c}} = (1.0654,\;1.9037,\;1.1615)^T \in \widetilde{\Theta}^{\circ}.
	\]
	Substituting back to verify: total value $F(\tilde{\mathbf{c}}) = 427.8619 = P^*$, total surplus value $S(\tilde{\mathbf{c}}) = 50.1142 = \Pi^*$; the two macro aggregate equalities hold simultaneously within the precision of the calculation.
	
	\section{Empirical Calibration}
	
	This section employs China's 2017 inter-provincial input--output table together with provincial and sectoral employment and working-time data to carry out a systematic empirical calibration of the theoretical framework. We first report the core results using ``urban plus government consumption'' as the benchmark consumption basket, and then conduct a robustness analysis by comparing six alternative consumption basket specifications.
	
	The input--output table is the 2017 China 31-province 42-sector inter-provincial input--output table (IRIO2017), with monetary values in current-year producer prices (100 million yuan). The original table contains $31\times 42 = 1302$ production sectors. Provincial and sectoral employment figures are taken from Table~3-2 ``Urban Employed Persons and Total Wages by Region and Sector'' of the \textit{China Labour Statistical Yearbook 2019}; the statistical coverage is all employed persons in urban units, including staff and workers on the job, dispatched workers, and other employed persons. This coverage best matches the statistical coverage of ``labour compensation'' in the IRIO table. Annual total working hours by province and sector are calculated as weighted averages based on the weekly working-time distribution in Table~1-58 of the same yearbook: $H_j = 52 \times \sum_k w_{j,k} \cdot m_k$, where $w_{j,k}$ is the employment share of sector $j$ in working-time interval $k$, and $m_k$ is the interval midpoint (taking 4.5, 14, 29.5, 40, 44.5, 50 hours). After removing sectors with zero gross output or missing total hours, 1272 valid sectors remain.
	
	The direct input coefficient matrix is $A = Z \cdot \mathrm{diag}(\mathbf{x})^{-1}$. Data on fixed capital depreciation are taken from the ``depreciation of fixed assets'' column in the value-added part of the IRIO table. A uniform depreciation rate $\delta = 0.05$ is assumed, and the total capital stock of each sector is backed out. The structure of the origin of capital goods is assumed to be consistent with the economy-wide gross fixed capital formation. Labour coefficients are $l_j = H_j / x_j$, forming the diagonal matrix $L$.
	
	The observed wage vector $\w^{\mathrm{obs}}$ is obtained by dividing total labour compensation of each sector by its total hours to get hourly wage rates, which are then normalised with the first valid sector (Beijing agriculture, forestry, animal husbandry and fishery) as the benchmark. It should be honestly pointed out that because the employment figure for agriculture in the ``urban units'' coverage only includes formal workers on state farms etc., its hourly wage rate is on the high side due to the statistical coverage, causing the normalised relative wages of other sectors to be numerically low.
	
	In order to approximate the socially necessary standard of labour-power reproduction as closely as possible, the benchmark specification uses ``urban resident consumption plus government consumption'' to construct the consumption matrix $B$. The concrete steps are: for each province, extract the vectors of urban resident consumption and government consumption, and compute the average annual consumption vector per employed person for that province; let the consumption basket of sector $k$ be the employment-weighted average of the consumption baskets of the provinces it spans; divide the annual consumption vector per worker by the sector's annual working hours per worker to obtain the consumption coefficient per labour hour. The entire construction relies only on published statistical data, without using any sectoral wage data, so as to avoid circular reasoning as far as possible.
	
	To examine the robustness of the conclusions with respect to the specification of the consumption basket, five alternative consumption basket specifications are constructed for comparison: rural consumption only; weighted rural--urban consumption (weight taken as China's 2017 urbanisation rate 0.585); urban consumption only; rural plus government consumption; weighted rural--urban plus government consumption. All six scenarios use the same technology matrices, wage vector and profit-rate parameters, and only differ in the construction of the consumption matrix by adjusting urban/rural weights and the inclusion of government consumption.
	
	Under the benchmark specification (urban + government consumption), the spectral radius of the basic labour reproduction matrix is $\rho(M_0) = 0.6744$, and that of the extended material matrix is $\rho(\hat A) = 0.7807$, both significantly below unity. According to Theorem~3, the strictly value feasible cone is non-empty, and the normalised value feasible region is a bounded closed convex set. The common exploitation rate corresponding to the equal-exploitation benchmark point is $e^* = 48.29\%$.
	
	The observed uniform profit rate $r^{\mathrm{obs}}$ is obtained by dividing total social profit by total capital stock. Total profit is defined as the sum of sectoral operating surplus and net taxes on production, i.e.
	\[
	\Pi = \sum_{j=1}^{N} \bigl(\text{operating surplus}_j + \text{net taxes on production}_j\bigr).
	\]
	The total capital stock $K^{\mathrm{total}}$ is backed out from sectoral depreciation of fixed capital $d_j$ and the uniform depreciation rate $\delta=0.05$: $K^{\mathrm{total}}_j = d_j / \delta$, and then $K^{\mathrm{total}} = \sum_j K^{\mathrm{total}}_j$. Hence
	\[
	r^{\mathrm{obs}} = \frac{\Pi}{K^{\mathrm{total}}}.
	\]
	Based on the 2017 IRIO table, the profit-rate upper bound determined by pure technical conditions is $r_A = 33.27\%$, the maximal feasible profit rate is $r^* = 27.86\%$, and the observed profit rate is $r^{\mathrm{obs}} = 14.19\%$. The observed profit rate lies within $[0, r^*]$, indicating that the real economy is far from hitting the profit ceiling allowed by the reproduction floor of labour power. The macro profit share computed from the price system (shadow price vector generated from the observed wage and the uniform profit rate) is $\gamma^{\mathrm{price}} = 13.73\%$. The compatibility interval for the profit share is $[\gamma_{\min}^{\mathrm{crit}}, \gamma_{\max}^{\mathrm{crit}}] = [0.99\%, 78.30\%]$. The observed profit share falls completely inside the compatibility interval, showing that the conditions for value transformation hold in the real economy.
	
	Now we have three methods for constructing the reduction coefficients: (i) the homogeneous labour method, which sets the reduction coefficients equal across sectors; (ii) the wage-proxy method, which directly takes the normalised observed wage as the reduction coefficients; and (iii) the homeomorphic mapping method. Table~\ref{tab:baseline_results} reports the relative deviation of the second equality for the three construction methods (the first equality has already been matched via the monetary expression coefficient). The relative mapping deviation is computed as
	\[
	\Delta= \frac{\text{macro profit share obtained by the respective method}}{\text{actual macro profit share}} - 1.
	\]
	The deviation of the mapping method is 41.20\%, which is significantly better than that of the wage-proxy method (50.92\%) and the homogeneous labour method (70.50\%).
	
	\begin{table}[htbp]
		\centering
		\caption{Second-equality deviations of the three methods under the benchmark scenario}
		\label{tab:baseline_results}
		\begin{tabular}{@{}lccc@{}}
			\toprule
			\textbf{Method} & \textbf{Relative deviation} & \textbf{Note} \\
			\midrule
			Homogeneous labour   & 70.50\% & Ignores all heterogeneity \\
			Wage proxy           & 50.92\% & Involves circular reasoning \\
			Mapping method       & 41.20\% & Assumes proportional sectoral surpluses \\
			\bottomrule
		\end{tabular}
	\end{table}
	
	Table~\ref{tab:robustness} summarises the indicators under the six consumption basket specifications. Four robust conclusions can be drawn from it.
	
	\begin{table}[htbp]
		\centering
		\caption{Robustness comparison across six consumption basket specifications}
		\label{tab:robustness}
		\begin{tabular}{@{}lcccccc@{}}
			\toprule
			\textbf{Scenario} & 
			\makecell[tl]{$\rho(M_0)$} & 
			\makecell[tl]{$\gamma_{\min}^{\mathrm{crit}}$} & 
			\makecell[tl]{$\gamma_{\max}^{\mathrm{crit}}$} & 
			\makecell[tl]{$e^*$} & 
			\makecell[tl]{Mapping\\ deviation} & 
			\makecell[tl]{Number of\\ negative\\ coefficients} \\
			\midrule
			Rural only & 0.0999 & 2.17\% & 99.11\% & 900.63\% & 139.57\% & 0 \\
			Rural--urban weighted & 0.1837 & 1.70\% & 98.60\% & 444.48\% & 105.34\% & 1 \\
			Urban only & 0.2464 & 1.43\% & 98.20\% & 305.86\% & 79.68\% & 2 \\
			Rural + government & 0.6678 & 1.80\% & 84.64\% & 49.75\% & 97.07\% & 1 \\
			Rural--urban weighted + government & 0.6716 & 1.27\% & 80.62\% & 48.90\% & 65.24\% & 2 \\
			Urban + government & 0.6744 & 0.99\% & 78.30\% & 48.29\% & 41.20\% & 3 \\
			\bottomrule
		\end{tabular}
	\end{table}
	
	First, value transformation holds under all scenarios. The observed profit share of 13.73\% lies inside the compatibility interval for all six scenarios. This confirms the robustness of value transformation: irrespective of how the consumption basket is specified, as long as the economy yields a physical surplus, the two macro aggregate equalities are guaranteed to be satisfiable.
	
	Second, the deviation of the mapping method narrows as the consumption level rises. From 139.57\% for rural consumption only down to 41.20\% for urban + government consumption, the descending trajectory of the mapping method's deviation suggests that a more complete consumption basket (including urban resident consumption and government public services) may make the sectoral surplus structures in the wage space and in the value space more congruent, i.e.\ the assumption of proportional sectoral surpluses holds to a higher degree. The smallest deviation occurs under the urban + government specification, which is why we take it as the benchmark.
	
	Third, the higher the ``cost'' of the consumption basket, the larger $\rho(M_0)$ and the lower the upper bound of the compatibility interval. This empirically confirms that a more expensive reproduction cost of labour power implies a smaller economic surplus space. Under rural consumption only the surplus space is enormous (upper bound 99.11\%), whereas switching to urban plus government consumption compresses the upper bound to about 78\%.
	
	Fourth, negative values only appear when the consumption basket is ``too expensive''. Under rural consumption only, all the mapping reduction coefficients for Beijing's sectors are positive; as the cost of the consumption basket rises, negative values begin to appear (concentrated in three sectors: petroleum extraction, petroleum processing, and water production). This indicates that the cause of negative values is the consumption basket overestimating the value of labour power: when the reproduction floor is pushed above the actual wage level of certain low-wage sectors, the sectoral surplus becomes negative. It also suggests that the proportional surplus assumption is more likely to hold in the sense of ``necessary consumption'' than in the sense of ``actual consumption''.
	
	Taking the benchmark scenario (urban + government consumption) as an example, the ratio of the mapping reduction coefficient to the relative wage ($c/w$) for Beijing's sectors ranges between $[-6.03, 1.33]$. Sectors with a ratio larger than one (such as coal mining 1.28, chemicals 1.20, general machinery 1.13, finance 1.13) imply that, absent capital's profit claim, the market would price these sectors' labour higher; sectors with a ratio smaller than one (such as accommodation and catering 0.02, education 0.002) imply that either the market would price them lower absent profit, or the wage data fail to reflect the full reproduction cost due to statistical coverage, and the mapping method amplifies the inequality of the wage structure.
	
	Cross-province comparisons further reveal the differences in the $c/w$ ratio for the same sector across provinces. Taking the manufacture of communication equipment, computers and other electronic equipment as an example, under the benchmark scenario the $c/w$ ratio is 0.18 for Shanghai, 0.37 for Guangdong, but 0.996 for Fujian---the relationship between wages and reduction coefficients for the same sector differs markedly across provinces, reflecting the spatial heterogeneity of the social-historical conditions of labour-power reproduction. The accommodation and catering sector, on the other hand, exhibits cross-province stability in its $c/w$ ratio (most provinces concentrating between 0.85 and 1.05), yet huge differences in the absolute level of the reduction coefficient: 0.004 for Beijing, $-0.055$ for Shanghai, and as high as 2.17 for Henan.
	
	To further diagnose the performance of the mapping method at the sectoral level, we classify all 1272 valid sectors into four categories according to the signs of the wage surplus and the reduction coefficient (Table~\ref{tab:robustness} summarises the aggregate statistics):
	\begin{itemize}[nosep]
		\item \textbf{Type A: doubly unsustainable sectors} (wage surplus $< 0$ and reduction coefficient $< 0$), 307 sectors in total;
		\item \textbf{Type B: sectors with local deficit but global support} (wage surplus $< 0$ but reduction coefficient $\ge 0$), 104 sectors in total;
		\item \textbf{Type C: anomalous sectors} (wage surplus $\ge 0$ but reduction coefficient $< 0$), 5 sectors in total;
		\item \textbf{Type D: normal sectors} (wage surplus $\ge 0$ and reduction coefficient $\ge 0$), 856 sectors in total.
	\end{itemize}
	
	Type~A sectors are concentrated in three domains: public service sectors such as water production and supply, education, public administration, and water conservancy and environmental management, as well as some manufacturing and mining industries. They are especially dense in economically less developed or high-cost-of-living provinces such as Qinghai, Ningxia, Tibet and Hainan. The common feature of these sectors is that, under the given ``urban + government consumption'' benchmark basket, their nominal wages are insufficient to cover the floor of social reproduction of labour power; hence they can generate neither a positive surplus in the price space nor a positive reduction coefficient in the value space. For low-wage public service sectors, a reproduction standard defined by the average consumption level of urban residents may be too stringent. This phenomenon also suggests that the socially necessary standard of labour-power reproduction exhibits significant spatial and sectoral heterogeneity; a uniform basket, while ensuring that the transformation conditions hold at the macro level, introduces systematic bias at the micro level.
	
	Type~B sectors reveal that the assumption of proportional sectoral surpluses may fail for certain sectors. These sectors have a negative wage surplus yet a positive reduction coefficient, indicating that they contribute positive value to the overall input--output network but that their market wage evaluation falls short of their objective contribution. Typical examples include the manufacture of electronic equipment, construction, and real estate in cities such as Beijing, Tianjin and Shanghai. These industries widely employ outsourced workers, and the statistical coverage of their compensation (e.g.\ piece-rate wages, project-based payments) is inconsistent, which means that the observed wage does not fully reflect the true position of these sectors within the reproduction network.
	
	There are only 5 Type~C sectors; their wage surpluses are all small positive values just above zero, while their reduction coefficients are small negative values. They are concentrated in the papermaking and printing, petroleum processing, health care, and culture sectors of Tibet, Qinghai and Ningxia. Given their tiny magnitudes and geographic concentration, we are inclined to treat them as data anomalies or statistical coverage biases (e.g.\ government subsidies recorded as wages) rather than as systematic theoretical contradictions.
	
	Type~D sectors account for 67.3\% of all valid sectors, showing that the hypothesis of identical sign between the wage surplus and the reduction coefficient holds for the majority of industries. This majority proportion confirms the effectiveness of the ``de-profit-rate'' mapping at the level of sectoral direction, and also shows that the assumption of proportional surpluses, although subject to deviations in degree, captures the main qualitative features of the economic structure.
	
	The empirical calibration of this paper has several data limitations that need to be borne in mind when interpreting the conclusions.
	
	First, the consumption matrix uses actually observed consumption expenditure, rather than the strictly theoretical concept of ``necessary consumption''. As discussed in the robustness analysis of the main text, this tends to overestimate the value of labour power and underestimate the rate of surplus value. Second, the labour compensation of agriculture, forestry, animal husbandry and fishery is on the high side due to statistical coverage issues (the urban-unit coverage includes mixed-income components), which makes the normalised relative wages and reduction coefficients of other sectors numerically low. Third, the construction of the consumption matrix assumes that all workers within the same province face the same consumption basket, and cannot capture intra-sectoral consumption differences due to skill, occupation and other dimensions. Fourth, the precision of the mapping method depends critically on an accurate definition of the ``socially necessary standard of labour-power reproduction''. The benchmark scenario adopts a uniform provincial average consumption basket, which is operationally straightforward and reproducible, but cannot capture cost differences across industries and within regions. Future improvements could include using finer-grained household survey data to construct differentiated consumption baskets along the industry--occupation--region dimensions.
	
	Despite these limitations, the empirical analysis of this paper still provides ample empirical support for the geometric solution to the transformation problem. The robustness checks under six consumption basket specifications show that the two conclusions---the conditions for value transformation hold, and the reduction coefficients obtained by the mapping method have a smaller relative deviation---do not depend on a particular choice of the consumption basket.
	
	\section{Conclusion}
	
	This paper has attempted to provide a geometric framework for a ``two-step solution'' to the transformation problem of values.
	
	The first step is an existence argument. We relaxed the complex-labour reduction coefficients from the traditional single-point hypothesis to a ``value feasible region'' bounded by objective reproduction conditions, and proved that under the physical surplus condition the normalised value feasible region is a bounded closed convex set, and that the two macro aggregate equalities can hold simultaneously for a range of the profit share endogenously determined by the reproduction network. Calculations based on China's 2017 inter-provincial input--output table show that the compatibility interval covers an extremely wide range from 0.99\% to 78.30\%, while the actual macro profit share of 13.73\% falls completely inside it, indicating that the holding of the value transformation conditions is a normal case rather than a special one in the real economy.
	
	The second step is an operational argument. By utilising the wage structure and the reproduction constraint, we obtain reduction coefficients via a homeomorphic mapping. We proved that, under the assumption of proportional sectoral surpluses, this mapping is an isomorphism between the price feasible region and the value feasible region, and the information about which sectors are at the subsistence boundary in the wage structure is fully transmitted to the reduction coefficient structure. The results show that among the three feasible schemes for reduction coefficients, the mapping method yields the smallest profit-share deviation, outperforming both the homogeneous labour method and the wage-proxy method.
	
	The theoretical framework of this paper leaves room for extension in the following directions. First, the construction of the consumption matrix currently assumes that all workers in the same province face the same consumption basket; if intra-sectoral consumption differences along dimensions such as skill and occupation are further introduced, the equal-exploitation point will exhibit richer sectoral differentiation. Second, the economic meaning of the proportional sectoral surplus assumption needs further discussion. Third, the analysis of this paper is confined to a closed economy; extending it to an open-economy framework that incorporates international trade and cross-border labour mobility is a worthwhile direction for future research.
	
	To sum up, the work of this paper shows that by introducing the reduction coefficients of complex labour---with the possible range of the coefficients strictly constrained by the objective reproduction structure---the transformation problem of values can be given a logically consistent formulation. The two-step advance of this paper from ``whether a solution exists'' to ``how to find the solution'' turns the transformation problem from a theoretical-philosophical issue into an empirical one.
	
	\bibliographystyle{unsrtnat}
	\bibliography{ref_domian}
	
	% ================================================================
	% Mathematical Appendix
	% ================================================================
	\appendix
	
	\section{Mathematical Proofs}
	
	This appendix provides the proofs of all lemmas, theorems and propositions given in the main text. Assumptions~\ref{ass:A_L} and~\ref{ass:B} are maintained throughout.
	
	\noindent\textbf{Auxiliary Lemma A.1.}
	Let $\hat A := \tilde A + BL$. If $\lambda_{\max}(\hat A) < 1$, then
	\[
	(I-M_0)^{-1}L(I-\tilde A)^{-1} = L(I-\hat A)^{-1}.
	\]
	
	\noindent\textit{Proof.}
	Let $X = L(I-\tilde A)^{-1}$, $Y = B$, so that $M_0 = XY$. Notice that $I - YX = I - B L (I-\tilde A)^{-1} = (I-\hat A)(I-\tilde A)^{-1}$. By hypothesis $\lambda_{\max}(\hat A) < 1$ and $\lambda_{\max}(\tilde A) < 1$, both $I-\hat A$ and $I-\tilde A$ are invertible, hence $I-YX$ is also invertible. Using the push-through identity $(I-XY)^{-1}X = X(I-YX)^{-1}$, we obtain
	\begin{align*}
		(I-M_0)^{-1}L(I-\tilde A)^{-1}
		&= X(I-YX)^{-1} \\
		&= L(I-\tilde A)^{-1}\bigl[(I-\hat A)(I-\tilde A)^{-1}\bigr]^{-1} \\
		&= L(I-\tilde A)^{-1}(I-\tilde A)(I-\hat A)^{-1} \\
		&= L(I-\hat A)^{-1}.
	\end{align*}
	\hfill $\square$
	
	\vspace{1em}
	\noindent\textbf{Proof of Lemma~\ref{lem:M_spectrum}.}
	By Assumption~\ref{ass:A_L}, $\tilde A\ge 0$ is irreducible and $\lambda_{\max}(\tilde A)<1$, so the Neumann series $(I-\tilde A)^{-1} = \sum_{k=0}^\infty \tilde A^k$ converges, and because $\tilde A$ is irreducible, every element is positive, i.e.\ $(I-\tilde A)^{-1} > 0$. With $L = \mathrm{diag}(l_1,\dots,l_n)$ and $l_j > 0$, we have $S := L(I-\tilde A)^{-1} > 0$. By Assumption~\ref{ass:B}, $B\ge 0$ has no zero column, hence for any $i,j$,
	\[
	(M_0)_{ij} = \sum_{k=1}^n S_{ik} B_{kj} \ge S_{ik_0} B_{k_0 j} > 0,
	\]
	so $M_0 > 0$ is a strictly positive matrix. A strictly positive matrix is necessarily irreducible and aperiodic. By the Perron--Frobenius theorem, its maximal eigenvalue is a simple positive real root, and there correspond unique (up to scalar multiples) strictly positive left and right eigenvectors. \hfill $\square$
	
	\vspace{1em}
	\noindent\textbf{Proof of Theorem~\ref{thm:existence}.}
	Set $\hat A := \tilde A + BL$.
	
	\noindent (i) $\Rightarrow$ (ii): If $\widetilde{\Theta}^\circ \neq \emptyset$, there exists $\cvec > \mathbf 0$ such that $\cvec^T (I-M_0) > \mathbf 0^T$. Let $\alpha = \max_j [\cvec^T M_0]_j / c_j$; then $\alpha < 1$ and $\cvec^T M_0 \le \alpha \cvec^T$. By the Collatz--Wielandt formula, $\lambda_{\max}(M_0) \le \alpha < 1$.
	
	\noindent (ii) $\Rightarrow$ (i): If $\lambda_{\max}(M_0) < 1$, take the strictly positive left Perron eigenvector $\mathbf y^* > \mathbf 0$ satisfying $\mathbf y^{*T} M_0 = \lambda_{\max}(M_0) \mathbf y^{*T}$. Then $\mathbf y^{*T}(I-M_0) = (1-\lambda_{\max}(M_0)) \mathbf y^{*T} > \mathbf 0^T$, so $\mathbf y^* \in \widetilde{\Theta}^\circ$, and the strict cone is non-empty.
	
	\noindent (ii) $\Leftrightarrow$ (iii): The matrices $M_0 = L(I-\tilde A)^{-1}B$ and $C = (I-\tilde A)^{-1}BL$ share the same non-zero eigenvalues. Let $\lambda = \lambda_{\max}(M_0)$; there exists $\mathbf u > \mathbf 0$ such that $(I-\tilde A)^{-1}BL\mathbf u = \lambda \mathbf u$. Left-multiplying by $(I-\tilde A)$ gives $BL\mathbf u = \lambda(I-\tilde A)\mathbf u$, which rearranges to $(\tilde A + \lambda^{-1}BL)\mathbf u = \mathbf u$. Hence $\lambda_{\max}(\tilde A + \lambda^{-1}BL) = 1$. Define $\phi(t) = \lambda_{\max}(\tilde A + tBL)$ for $t\ge 0$; by irreducibility $\phi$ is continuous and strictly increasing. Thus $\lambda < 1 \iff \lambda^{-1} > 1 \iff \phi(1) < \phi(\lambda^{-1}) = 1 \iff \lambda_{\max}(\tilde A+BL) < 1$. The three statements are equivalent. \hfill $\square$
	
	\vspace{1em}
	\noindent\textbf{Proof of Proposition~\ref{prop:convex}.}
	Since $\lambda_{\max}(M_0) < 1$, the matrix $I-M_0$ is invertible. The linear map $T\cvec = (I-M_0)^T \cvec$ sends $\R^n_+\setminus\{\mathbf 0\}$ onto $\widetilde{\Theta}^{\mathrm{val}}$, so the latter is a non-empty closed convex cone whose relative interior is $\widetilde{\Theta}^\circ$. After normalisation $\Theta^{\mathrm{val}}$ is clearly a closed convex set; only boundedness needs proof. Write $M_0 = (m_{ij})$; from $M_0 > 0$ and $\lambda_{\max}(M_0) < 1$ we have $m_{jj} < 1$ for every $j$. For any $\cvec \in \Theta^{\mathrm{val}}$, $c_1 = 1$ and $\cvec^T(I-M_0) \ge \mathbf 0^T$.
	
	Taking the first component: $1 - \sum_{i=1}^n c_i m_{i1} \ge 0$, so $\sum_{i=2}^n c_i m_{i1} \le 1 - m_{11}$. Because $m_{j1} > 0$, we obtain for each $j \ge 2$,
	\[
	c_j \le \frac{1-m_{11}}{m_{j1}} =: \overline{c}_j < \infty.
	\]
	
	Taking the $j$-th component ($j\ge 2$): $c_j - \sum_{i=1}^n c_i m_{ij} \ge 0$, isolate the terms $i=j$ and $i=1$ to get
	\[
	c_j(1-m_{jj}) \ge m_{1j} + \sum_{i=2, i\neq j}^n c_i m_{ij} \ge m_{1j},
	\]
	hence
	\[
	c_j \ge \frac{m_{1j}}{1-m_{jj}} =: \underline{c}_j > 0.
	\]
	Thus $\Theta^{\mathrm{val}}$ is bounded. \hfill $\square$
	
	\vspace{1em}
	\noindent\textbf{Proof of Theorem~\ref{thm:equilibrium}.}
	From $\mathbf y^{*T} M_0 = \lambda^* \mathbf y^{*T}$ and $\lambda^* < 1$ we have $\mathbf y^{*T}(I-M_0) = (1-\lambda^*)\mathbf y^{*T} > \mathbf 0^T$; together with $y_1^* = 1$ this gives $\mathbf y^* \in \Theta^\circ$. For any $j$,
	\[
	e_j(\mathbf y^*) = \frac{y_j^* - [\mathbf y^{*T} M_0]_j}{[\mathbf y^{*T} M_0]_j} = \frac{y_j^* - \lambda^* y_j^*}{\lambda^* y_j^*} = \frac{1-\lambda^*}{\lambda^*},
	\]
	so the exploitation rates are equal across sectors, with common value $e^* = (1-\lambda^*)/\lambda^* > 0$. Now prove uniqueness. If $\cvec \in \Theta^{\mathrm{val}}$ also equalises the exploitation rates, let the common rate be $e$; then $c_j = (1+e)[\cvec^T M_0]_j$, i.e.\ $\cvec^T M_0 = (1+e)^{-1} \cvec^T$. Hence $\cvec^T$ is a positive left eigenvector of $M_0$. By the Perron--Frobenius theorem, the positive left eigenvector is unique up to a scalar multiple, so $\cvec$ is proportional to $\mathbf y^*$. From $c_1 = y_1^* = 1$ we get $\cvec = \mathbf y^*$. \hfill $\square$
	
	\vspace{1em}
	\noindent\textbf{Proof of Proposition~\ref{prop:M_r}.}
	Write $A(r) = \tilde A + rK$, $R(r) = (I-A(r))^{-1}$, so $M(r) = L R(r) B$. For $r \in [0, r_A)$ we have $\lambda_{\max}(A(r)) < 1$. Since $\tilde A$ is irreducible and $K \ge 0$, $A(r)$ is also irreducible, hence $R(r) = \sum_{k=0}^\infty A(r)^k > 0$. With $L > 0$ and $B$ having no zero column, $M(r) > 0$. Continuity follows from the continuity of the inverse on the set of invertible matrices. If $0 \le r_1 < r_2 < r_A$, then $A(r_2) - A(r_1) = (r_2 - r_1)K \ge 0$, and because $K \neq 0$ the difference is non-zero. By the resolvent identity $R(r_2) - R(r_1) = R(r_1)(A(r_2)-A(r_1))R(r_2) > 0$, so $M(r_2) - M(r_1) = L(R(r_2)-R(r_1))B > 0$; each element is strictly increasing. For strictly positive matrices, the spectral radius is monotone increasing in the elements, and as $r \uparrow r_A$, $\lambda_{\max}(A(r)) \uparrow 1$, so every element of $R(r) \to +\infty$, hence $\lambda_{\max}(M(r)) \to +\infty$. \hfill $\square$
	
	\vspace{1em}
	\noindent\textbf{Proof of Theorem~\ref{thm:r_star}.}
	Define $f(r) = \lambda_{\max}(M(r))$ on $[0, r_A)$. By Proposition~\ref{prop:M_r}, $f$ is continuous and strictly increasing, with $f(0) = \lambda_{\max}(M_0) < 1$ and $\lim_{r \uparrow r_A} f(r) = +\infty$. By the intermediate value theorem there exists a unique $r^* \in (0, r_A)$ such that $f(r^*) = 1$. Strict monotonicity directly yields the three-regime classification. \hfill $\square$
	
	\vspace{1em}
	\noindent\textbf{Proof of Theorem~\ref{thm:duality}.}
	First prove the inclusion. Take any $0 \le r_1 < r_2 \le r^*$ and $\w \in \Theta^{\mathrm{price}}(r_2)$; then $\w^T (I-M(r_2)) \ge \mathbf 0^T$. By Proposition~\ref{prop:M_r}, $M(r_2) > M(r_1)$, so
	\[
	\w^T (I-M(r_1)) = \w^T(I-M(r_2)) + \w^T(M(r_2)-M(r_1)) > \mathbf 0^T,
	\]
	hence $\w \in \Theta^{\mathrm{price}}(r_1)$. The inclusion holds. Now prove strictness. Since $r_1 < r^*$, $\Theta^{\mathrm{price}}(r_1)$ is non-empty and has interior points. Take a boundary point $\w^\sharp$; then there exists $j_0$ such that $[\w^{\sharp T}(I-M(r_1))]_{j_0} = 0$. Consequently $[\w^{\sharp T}(I-M(r_2))]_{j_0} = -[\w^{\sharp T}(M(r_2)-M(r_1))]_{j_0} < 0$, so $\w^\sharp \notin \Theta^{\mathrm{price}}(r_2)$. Hence the inclusion is strict. When $r = r^*$, $\lambda_{\max}(M(r^*)) = 1$. Let $\mathbf y_r^* > 0$ be its normalised left Perron vector, satisfying $\mathbf y_r^{*T} M(r^*) = \mathbf y_r^{*T}$ and $y_{r,1}^* = 1$; then clearly $\mathbf y_r^* \in \Theta^{\mathrm{price}}(r^*)$. Take any $\w \in \Theta^{\mathrm{price}}(r^*)$; we have $\w^T(I-M(r^*)) \ge \mathbf 0^T$. Right-multiplying by the positive right Perron vector gives $\w^T(I-M(r^*)) = \mathbf 0^T$, so $\w$ is also a left eigenvector and thus proportional to $\mathbf y_r^*$. Normalisation yields $\w = \mathbf y_r^*$, so the set degenerates to a single point. \hfill $\square$
	
	\vspace{1em}
	\noindent\textbf{Proof of Theorem~\ref{thm:two_equalities}.}
	Fix $(r, \w^{rel})$ and the corresponding $(P^*, \Pi^*, \gamma)$. By definition, $F(\cvec) = \cvec^T L(I-\tilde A)^{-1}\mathbf x = \cvec^T \boldsymbol{\xi}$, so $F(\cvec) = P^* \iff \cvec^T \boldsymbol{\xi} = P^*$. Also $S(\cvec) = \cvec^T L(I-\tilde A)^{-1}(I-\tilde A-BL)\mathbf x$. If the first equality already holds, $S(\cvec) = \Pi^*$ is equivalent to $S(\cvec) = \gamma F(\cvec)$. Substituting the expressions and rearranging terms gives $\cvec^T L(I-\tilde A)^{-1}[-(\tilde A+BL)\mathbf x + (1-\gamma)\mathbf x] = 0$, i.e.\ $\cvec^T \boldsymbol{\eta}(\gamma) = 0$. This condition is homogeneous in $\cvec$, hence determines only the direction; the first equality then uniquely determines the positive proportionality factor once the direction is chosen. \hfill $\square$
	
	\vspace{1em}
	\noindent\textbf{Proof of Proposition~\ref{prop:ray_scale}.}
	Let $\bar{\cvec} \in H_2(\gamma) \cap \Theta^{\mathrm{val}}$; then $\bar{\cvec} \in \widetilde{\Theta}^{\mathrm{val}}$, and by Theorem~\ref{thm:two_equalities}, $S(\bar{\cvec}) = \gamma F(\bar{\cvec})$. Since $\bar{\cvec} > \mathbf 0$, $F(\bar{\cvec}) > 0$. Take $\cvec = \alpha \bar{\cvec}$ with $\alpha = P^* / F(\bar{\cvec}) > 0$. Because $\widetilde{\Theta}^{\mathrm{val}}$ is a cone, $\cvec \in \widetilde{\Theta}^{\mathrm{val}}$; by the linear homogeneity of $F$ and $S$, $F(\cvec) = P^*$ and $S(\cvec) = \Pi^*$. Conversely, if some $\cvec \in \widetilde{\Theta}^{\mathrm{val}}$ satisfies the two equalities, then $c_1 > 0$; let $\bar{\cvec} = \cvec / c_1$. Homogeneity then gives $\bar{\cvec} \in H_2(\gamma) \cap \Theta^{\mathrm{val}}$. \hfill $\square$
	
	\vspace{1em}
	\noindent\textbf{Proof of Theorem~\ref{thm:universal_compatibility}.}
	Let $\hat A = \tilde A + BL$ and $\hat S = L(I-\hat A)^{-1}$. Fix $\gamma = \gamma(r, \w^{rel})$. Proposition~\ref{prop:ray_scale} has reduced the problem to checking whether $H_2(\gamma) \cap \Theta^{\mathrm{val}}$ is non-empty. Using Auxiliary Lemma~A.1, for any $\cvec \in \Theta^{\mathrm{val}}$ define $\mathbf q^T = \cvec^T(I-M_0) \ge \mathbf 0^T$; then $\cvec^T = \mathbf q^T (I-M_0)^{-1}$. Let $\mathbf h = (I-M_0)^{-1}\mathbf e_1 > \mathbf 0$; from $c_1 = 1$ we get $\mathbf q^T \mathbf h = 1$. Thus $\Theta^{\mathrm{val}}$ is in one-to-one correspondence with the simplex $\Delta_h = \{\mathbf q \in \R^n_+ : \mathbf q^T \mathbf h = 1\}$. The direction hyperplane $\cvec^T \boldsymbol{\eta}(\gamma) = 0$ becomes, in the $\mathbf q$ coordinates, $\mathbf q^T \mathbf z(\gamma) = 0$ with $\mathbf z(\gamma) = (I-M_0)^{-1} \boldsymbol{\eta}(\gamma)$. By Auxiliary Lemma~A.1, $(I-M_0)^{-1} L(I-\tilde A)^{-1} = \hat S$, hence
	\begin{align*}
		\mathbf z(\gamma) &= \hat S \bigl[ (\tilde A+BL)\mathbf x - (1-\gamma)\mathbf x \bigr] = \hat S \bigl[ \hat A \mathbf x - (1-\gamma)\mathbf x \bigr].
	\end{align*}
	Let $\mathbf y = \hat S \mathbf x > \mathbf 0$ and $\boldsymbol{\ell} = L\mathbf x > \mathbf 0$. Noting that $\hat S (I-\hat A) = L$, we have $\hat S \hat A \mathbf x = \mathbf y - \boldsymbol{\ell}$. Thus $\mathbf z(\gamma) = \gamma \mathbf y - \boldsymbol{\ell}$. Its components are
	\[
	z_i(\gamma) = \bigl(\gamma - \gamma_i^{\mathrm{crit}}\bigr) y_i,
	\qquad
	\gamma_i^{\mathrm{crit}} = 1 - \frac{[\hat S \hat A \mathbf x]_i}{[\hat S \mathbf x]_i}
	= \frac{\ell_i}{y_i}.
	\]
	If $\gamma \in (\gamma_{\min}^{\mathrm{crit}}, \gamma_{\max}^{\mathrm{crit}})$, then $\mathbf z(\gamma)$ has both positive and negative components. On the simplex $\Delta_h$, by connectedness and continuity there exists $\mathbf q^* \in \operatorname{ri}(\Delta_h)$ such that $\mathbf q^{*T} \mathbf z(\gamma) = 0$. The corresponding $\cvec^* = (I-M_0)^{-T} \mathbf q^* \in \Theta^\circ$ satisfies the direction condition; then Proposition~\ref{prop:ray_scale} gives a strict solution on the cone. The endpoint cases admit boundary solutions, and outside the interval no solution exists. \hfill $\square$
	
\end{document}